%% file: root.tex
\documentclass[letterpaper, 10 pt, conference]{ieeeconf}
\IEEEoverridecommandlockouts
\usepackage[utf8]{inputenc} 
\usepackage{dsfont}

\usepackage{amsfonts}
\usepackage{amsmath}
\usepackage{amssymb}
\usepackage{hyperref}       
\usepackage{url}            
\usepackage{booktabs}       

\usepackage{nicefrac}       
\usepackage{microtype}      
\usepackage{algorithmic}
\usepackage{mathrsfs}
\usepackage{graphicx}
\usepackage{textcomp}
\usepackage{xcolor}
\usepackage{subcaption}

\newtheorem{corollary}{Corollary}
\newtheorem{lemma}{Lemma}

\newtheorem{assumption*}{Assumption}
\newtheorem{stdassumption*}{Standing Assumption}

\newtheorem{proposition}{Proposition}
\newtheorem{definition}{Definition}
\newtheorem{definition*}{Definition}

\usepackage{mathtools}
\graphicspath{ {./images/} }

\usepackage[ruled,vlined]{algorithm2e}

\usepackage{cleveref}

\def\BibTeX{{\rm B\kern-.05em{\sc i\kern-.025em b}\kern-.08em
    T\kern-.1667em\lower.7ex\hbox{E}\kern-.125emX}}
    
\begin{document}

\title{\bf \LARGE 
  Differential Privacy in Generative AI Agents:\\ Analysis and Optimal Tradeoffs
}

\author{Ya-Ting Yang and Quanyan Zhu
\thanks{Authors are with the Department of Electrical and Computer Engineering, New York University, NY, 11201, USA {\tt \{yy4348, qz494\}@nyu.edu}}
}

\maketitle

\input{abstract}

\input{introduction}

\input{relatedwork}

\input{model}

\input{DP}

\input{problem}

\input{experiment}

\input{conclusion}

\bibliographystyle{abbrv}
\bibliography{ref}

\end{document}

%% file: abstract.tex
\begin{abstract}
Large language models (LLMs) and AI agents are increasingly integrated into enterprise systems to access internal databases and generate context-aware responses. While such integration improves productivity and decision support, the model outputs may inadvertently reveal sensitive information. Although many prior efforts focus on protecting the privacy of user prompts, relatively few studies consider privacy risks from the enterprise data perspective. Hence, this paper develops a probabilistic framework for analyzing privacy leakage in AI agents based on differential privacy. We model response generation as a stochastic mechanism that maps prompts and datasets to distributions over token sequences. Within this framework, we introduce token-level and message-level differential privacy and derive privacy bounds that relate privacy leakage to generation parameters such as temperature and message length. We further formulate a privacy–utility design problem that characterizes optimal temperature selection. 

\end{abstract}

%% file: introduction.tex
\section{Introduction}
\label{sec:intro}

Enterprises are increasingly deploying large language models (LLMs) and AI agents to support decision making, automate workflows, and enable natural-language access to internal knowledge bases \cite{10852443,11022699}. In many deployments, these systems are connected to proprietary databases, internal documents, and operational data sources in order to generate context-aware and informative responses. While such integration significantly enhances the functionality and efficiency of enterprise systems, it also introduces important privacy and security concerns \cite{10.1145/3773080,11081880,yan2025protecting,das2025security}.

A central risk arises from potential information leakage through model outputs. When an LLM generates responses conditioned on internal data, the output may inadvertently reveal sensitive information contained in the underlying database \cite{11062758}. Even when confidential records are not directly disclosed, adversaries may infer sensitive attributes by strategically crafting queries and analyzing the model’s responses. Through repeated interactions, attackers can gradually reconstruct hidden information about individuals, business operations, or proprietary assets. Such inference-based leakage is particularly concerning because it may occur indirectly through seemingly benign responses.

To mitigate these risks, existing enterprise solutions largely rely on guardrail-based protection mechanisms \cite{10916629}. These approaches typically include prompt filtering, rule-based detection, sanitization, retrieval restrictions, and post-generation content moderation \cite{11409403,11296166,zhang2025dyntext}. Although such guardrails provide a practical layer of defense, they are often implemented as ad hoc mechanisms tailored to specific prompts or sensitive attributes \cite{yue2021differential,thareja2025dp}. Consequently, they lack a principled theoretical foundation and provide limited guarantees regarding the level of privacy protection achieved.

Another challenge is the rarely few well-defined metrics for quantifying privacy leakage in LLM responses. Without formal measures of how much information can be inferred from generated outputs, it becomes difficult to evaluate the effectiveness of guardrails or compare different privacy protection strategies. Current industry practices often rely on empirical, heuristic, or manual evaluations or manual, which may not scale well to complex enterprise environments \cite{11401523}.

The problem becomes even more challenging in dynamic data environments. In many enterprise systems, databases are continuously updated as new records are added, modified, or removed. Each change in the underlying data may alter the information encoded in an LLM's responses. As a result, guardrails designed for a previous database state may quickly become outdated. Maintaining accurate and robust privacy protection mechanisms in such evolving environments is therefore operationally difficult and may leave temporary windows of vulnerability.


This paper addresses these challenges by developing a systematic framework for analyzing privacy leakage in LLM agents from a probabilistic perspective. The contributions of this work are fourfold. First, we propose an input–output system-level model that represents an LLM agent as a stochastic generative mechanism mapping prompts, contextual information, and datasets to distributions over messages. Second, we characterize inference-based privacy leakage at both the token and message levels. Third, we derive explicit privacy bounds that relate privacy leakage to generation parameters such as message length and temperature. Finally, we study an optimal privacy–utility design problem for LLM agents, showing how system parameters can be tuned to balance response quality and privacy protection.

%% file: relatedwork.tex
\section{Literature Review}

Prior research has studied privacy leakage and  disclosure of sensitive information problems in large language models (LLMs) from several complementary perspectives.

\subsection{Enterprise Guardrails and Secure Architectures}

Empirical studies have examined the vulnerability of LLM systems when exposed proprietary enterprise data. The study in \cite{11398712} evaluates whether LLMs retain and disclose confidential organizational information such as HR records, financial documents, and source code. Their experiments show that models can leak sensitive information under prompt injection or role-based manipulation attacks, particularly in the context of GDPR and intellectual property, highlighting the need for stronger privacy protection in enterprise deployments.

Many proposed enterprise deployments rely on guardrail-based security mechanisms to reduce the risk of data leakage. These systems typically introduce intermediate components that sanitize inputs or constrain outputs when interacting with LLM services. For example, \cite{10916629} proposes a privacy-focused gatekeeper LLM that pseudonymizes personally identifiable information (PII) and assigns unique identifiers as new mapping. Similarly, Portcullis \cite{zhan2025portcullis} acts as a privacy gateway that anonymizes sensitive data through parallel substitution, securely communicates with other services, and reconstructs responses within encrypted memory. SeLLMA \cite{wacker2025sellma} then employs advanced PII detection and secure data transformation techniques to safeguard sensitive information during LLM interactions. 

Another line of efforts focus on designing secure deployment architectures or secure pipelines for enterprise environments. The framework in \cite{11296166} introduces a modular system incorporating input sanitization, compliance policy enforcement, role-based access control, secure data flows, and audit logging to satisfy regulatory requirements such as HIPAA and GDPR. Similarly, SHIELD \cite{narajala2025securing} proposes mitigation strategies for organizations such as monitoring and segmentation tailored for different conditions. PrivacyAsst \cite{10458329} proposes a privacy-preserving framework specifically for tool-using LLM agents that combines homomorphic encryption with attribute-shuffling mechanisms to conceal user inputs while preserving functionality.

\subsection{Differential Privacy for Language Models}

Differential privacy (DP) has emerged as one of the most widely studied mechanisms for providing formal privacy guarantees in machine learning. Several works adapt DP mechanisms to LLM training, fine-tuning, and inference processes. For example, \cite{majmudar2022differentially} proposes a lightweight perturbation mechanism applied during the decoding stage of trained models to avoid the significant computational cost of retraining privacy-preserving LLMs from scratch. Similarly, \cite{chen2025protecting} introduces Privacy-Flat, a framework that improves the privacy-utility trade-off in DP-SGD-trained models by controlling weight flatness.

Several works also introduce adaptive token-level privacy mechanisms. For example, \cite{yue2021differential} partitions text into sensitive and non-sensitive token sets and allocates different privacy budgets to each group. DYNTEXT \cite{zhang2025dyntext} adjusts the protection level of tokens based on their semantic sensitivity to balance the privacy–utility trade-off. Similarly, \cite{11409403} performs risk identification and personalized labeling before applying risk-aware LDP sanitization mechanisms for both token-level and sentence-level privacy. DP-Fusion \cite{thareja2025dp} proposes a DP inference mechanism that blends the output distributions of models with and without sensitive tokens, ensuring that the influence of sensitive context tokens in the generated output remains bounded.

Other studies investigate DP mechanisms for specific data modalities. PrivCode \cite{liu2025privcode} proposes a two-stage framework for generating differentially private synthetic code datasets while preserving code structure and utility. Meanwhile, \cite{wang2026privacy} studies privacy trading mechanisms in federated LLM fine-tuning using a game-theoretic framework. When it comes to retrieval-augmented generation (RAG), \cite{koga2024privacy} explores differentially private retrieval mechanisms and allocates privacy budgets to selected tokens with sensitive information.




%% file: model.tex
\section{Generative Message Model}

We consider a system composed of $N$ interacting agents indexed by $i \in \mathcal{N} = \{1,\ldots,N\}$. At each time step $t = 0,1,2, \ldots$, agents can communicate by producing messages that influence the information available to other agents to shape the evolution of the interaction over time. The objective of this model is to describe how agents based on language models generate responses and how their internal reasoning states can be represented probabilistically. We assume that all random variables are defined  on a probability space $(\Omega,\mathcal{F},\mathbb{P})$ throughout the analysis.

\subsection{Token and Message Spaces}

Each agent $i$ operates using a token vocabulary $\mathcal{V}_i$. In practice, tokens may correspond to words, subwords, or characters, depending on the model's tokenization scheme. A message produced by the language model is therefore a finite sequence of tokens drawn from this vocabulary.

\begin{definition}[Message Space]
Let $\mathcal{V}_i$ denote the token vocabulary of agent $i$. For each integer $L \ge 1$, let $\mathcal{V}_i^L$ denote the set of all token sequences of length $L$, that is, $\mathcal{V}_i^L=\{(w_1,\ldots,w_L) : w_k \in \mathcal{V}_i \text{ for } k=1,\ldots,L\}$. The message space of agent $i$ is defined as $\mathcal{X}_i = \bigcup_{L=1}^{\infty} \mathcal{V}_i^{L}$.
\end{definition}

Although natural language can generate many possible sentences, this representation allows us to treat each message as an element of a well-defined mathematical space. For mathematical analysis, we equip $\mathcal{X}_i$ with the discrete $\sigma$-algebra. A message generated by agent $i$ at time $t$ is therefore represented as a random variable $m_i^t : \Omega \to \mathcal{X}_i$.

\subsection{Prompts and Information Sets}
The message generated by the language model depends on contextual inputs available to the agent. First, each agent receives a prompt $p_i^t$ that initiates the generation process. We assume that prompts belong to a prompt space $\mathcal{P}_i$, so that  $p_i^t \in \mathcal{P}_i$. Second, the agent has access to an information set that summarizes all contextual information available when generating the response.


A general representation of the information set is $I_i^t =\left(m_1^{0:t-1},\ldots,m_N^{0:t-1}, M^t, D_i^t, o_i^t\right)$. Here $m_j^{0:t-1}=(m_j^0,\ldots,m_j^{t-1})$ denotes the sequence of messages previously generated by agent $j$. The variable $M^t$ denotes shared memory that may be accessible to all agents, such as shared conversation history. The variable $D_i^t$ represents documents retrieved by agent $i$ from external sources, such as retrieval-augmented generation systems. Finally, $o_i^t$ denotes observations obtained from the surrounding environment. The pair $(p_i^t,I_i^t)$ therefore represents the epistemic state of agent $i$ at time $t$.

\subsection{Token-Level Generative Mechanism}

Modern language models generate text by sampling tokens sequentially from probability distributions produced by the model. A message generated by agent $i$ at time $t$ is a finite token sequence  $m_i^t=(w_{i,1},\ldots,w_{i,L_i^t})$, where each token satisfies $w_{i,k}\in\mathcal{V}_i$ and $L_i^t$ denotes the length of the message. At each position $k$, the probability of the next token depends on the previously generated tokens and the contextual inputs available to the agent. Formally, the language model defines conditional token probabilities $\pi_i^\theta(w_{i,k} | w_{i,1:k-1}, p_i^t, I_i^t)$, where $p_i^t$ denotes the prompt, $I_i^t$ denotes the information set available to the agent, and $\theta$ denotes the internal parameters of the model.
The probability of generating the token sequence 
$m_i^t=(w_{i,1},\ldots,w_{i,L_i^t})$ therefore satisfies the autoregressive factorization
\begin{equation}
\mathbb{P}_\theta(m_i^t | p_i^t,I_i^t)=
\prod_{k=1}^{L_i^t}
\pi_i^\theta(w_{i,k} | w_{i,1:k-1}, p_i^t, I_i^t).
\label{eq:autoregressive_factorization}
\end{equation}

The token-level generation process induces a probability distribution over the message space $\mathcal{X}_i$. For any message $x=(w_1,\ldots,w_L)\in\mathcal{X}_i$, the probability assigned to that message is given by
\begin{equation}
Q_i^t(x)
=
\mathbb{P}_\theta(x \mid p_i^t,I_i^t). 
\label{eq:induced_message_distribution}
\end{equation}
The message produced by the agent is obtained by sampling $m_i^t \sim Q_i^t$. Hence response generation can be viewed as a two-stage mechanism: contextual inputs determine a distribution over messages, and a specific message is then sampled from that distribution.
%
%
For many modern language models the token-level probabilities are explicitly observable. At each generation step, the model produces a vector of logits over the vocabulary $\mathcal{V}_i$. Applying the softmax transformation to these logits yields the conditional token probabilities $\pi_i^\theta(\cdot)$ used in 
\eqref{eq:autoregressive_factorization} and 
\eqref{eq:induced_message_distribution}.

\subsection{Temperature and Length of the Message}

The internal parameter vector $\theta_i$ influences the geometry of the token generation process. Two particularly important control parameters are the decoding temperature $T_i$ and the length of the generated message $L_i^t$. These parameters affect the dispersion of the token distribution and the complexity of the resulting message sequence.

A key component of the token generation mechanism is the logit function produced by the language model. Let $\ell_i : \mathcal{V}_i \times \mathcal{V}_i^{*} \times \mathcal{P}_i \times \mathcal{I}_i \mapsto \mathbb{R},$ denote the logit mapping associated with agent $i$, where $\ell_i(w \mid w_{1:k-1},p_i^t,I_i^t)$ assigns a real-valued score to token $w \in \mathcal{V}_i$ given the previously generated token sequence $w_{1:k-1}$ and the contextual inputs $(p_i^t,I_i^t)$. These scores represent the model's relative preference for different tokens.

The logits are converted into probabilities through a temperature-scaled softmax transformation. The conditional token distribution is therefore defined by
$$\pi_i^\theta(w | w_{1:k-1},p_i^t,I_i^t)=\frac{\exp\!\left(\ell_i(w | w_{1:k-1},p_i^t,I_i^t)/T_i\right)}{\sum_{v\in\mathcal{V}_i}\exp\!\left(\ell_i(v | w_{1:k-1},p_i^t,I_i^t)/T_i\right)},$$
where $T_i>0$ denotes the temperature parameter. The temperature parameter determines the dispersion of the token distribution. As $T_i \to 0$, the distribution concentrates on tokens with the highest logit values, and the model approaches deterministic decoding. As $T_i \to \infty$, the distribution approaches the uniform distribution over $\mathcal{V}_i$. 

The second important control parameter is the message length $L_i^t$, which determines the number of tokens generated in the response. The length parameter affects the size of the sequence space and, therefore, the complexity of the induced message distribution. Consequently, longer messages correspond to higher-dimensional sequence spaces and may lead to more diverse semantic outcomes.

%% file: DP.tex
\section{Differential Privacy of LLM-Agent}
In this section, we study the privacy properties of the response-generation mechanism of a language model agent. Consider a single LLM agent $i$. At time $t$ the agent receives inputs $(p_i^t, D_i^t, I_i^t)$, where $p_i^t$ denotes the prompt, $D_i^t$ denotes the dataset accessible to the agent; with a slight abuse of notation, we let $I_i^t$ encompass contextual information excluding $D_i^t$. The agent produces a response $m_i^t$ that belongs to the message space $\mathcal X_i$ defined earlier. Because the generation process is stochastic, the output message is drawn from a probability distribution over $\mathcal X_i$, that is, $m_i^t \sim Q_i^t(\cdot \mid p_i^t, D_i^t, I_i^t)$. Equivalently, the response generation process can be viewed as a randomized mechanism $\mathcal M_i : (p_i^t, D_i^t, I_i^t) \mapsto \mathcal X_i$, which maps the inputs to a random message in the message space. 

\subsection{Message-Level Differential Privacy}

Differential privacy formalizes the requirement that the output of a mechanism should not depend strongly on any single record in the dataset. Let $\mathcal D$ denote the space of all possible datasets. A dataset $D \in \mathcal D$ is a finite collection of records $D=(r_1,\ldots,r_n)$ where each record $r_j$ belongs to a data domain $\mathcal R$. Datasets $D,D' \in \mathcal D$ are said to be \emph{neighboring} ($D \sim D'$), if they differ in exactly one record.
\begin{definition}[Message Differential Privacy]
The generation mechanism $\mathcal M_i$ satisfies
$(\varepsilon,\delta)$-DP if for every neighboring datasets $D\sim D'$, for all prompts $p_i^t$, and for every measurable output set $S \subseteq \mathcal X_i$,
\begin{equation*}
    \mathbb P(m_i^t \in S | p_i^t, D, I_i^t)\le e^{\varepsilon} \mathbb P(m_i^t \in S | p_i^t, D', I_i^t)+\delta.
\end{equation*}
\label{def:message_DP}
\vspace{-3mm}
\end{definition}

\subsection{Token-Level Differential Privacy}

A generated message $m_i^t$ can be written as a token sequence $m_i^t=(w_{i,1},\ldots,w_{i,L_i^t})$, Let
$h_k=(w_{i,1:k-1},p_i^t,I_i^t)$ denote the token history at step $k$, the conditional token distribution is written as $\pi_{i,D}^\theta(w_{i,k}\mid h_k)$.

\begin{definition}[Token Differential Privacy]
The token generation mechanism satisfies $(\varepsilon_k,\delta_k)$-DP at step $k$ if for every neighboring datasets $D \sim D'$, for all token histories $h_k$, and for all tokens $w \in \mathcal V_i$,
\begin{equation*}
    \pi_{i,D}^\theta(w\mid h_k)\le e^{\varepsilon_k} \pi_{i,D'}^\theta(w\mid h_k) + \delta_k.
\end{equation*}
\label{def:token_DP}
\end{definition}
\begin{lemma}
If the token generation mechanism at each step $k$ satisfies $(\varepsilon_k,\delta_k)$-DP in Definition \ref{def:token_DP}, then the induced message-generation mechanism $\mathcal M_i : (p_i^t,D_i^t,I_i^t) \to \mathcal X_i$ satisfies $(\varepsilon,\delta)$-DP in Definition \ref{def:message_DP} at the message level with $\varepsilon = \sum_{k=1}^{L_i^t}\varepsilon_k, \delta = \sum_{k=1}^{L_i^t}\delta_k$.
\label{lem:DP_sum}
\end{lemma}
\begin{proof}
    According to \eqref{eq:autoregressive_factorization}, the generation process can therefore be viewed as a sequential composition of $L_i^t$ randomized mechanisms. By the composition property of DP, the privacy loss parameters accumulate additively across sequential mechanisms. Consequently, the overall response generation mechanism satisfies $(\varepsilon,\delta)$-DP with $\varepsilon=\sum_{k=1}^{L_i^t}\varepsilon_k$ and $\delta=\sum_{k=1}^{L_i^t}\delta_k$.
\end{proof}

\begin{proposition}[Token-Level Privacy Bound]\label{thm:temp-token-privacy}
Assume that the logit mapping $\ell_i(w,h)$ satisfies $\sup_{w,h} |\ell_{i,D}(w,h)-\ell_{i,D'}(w,h)| \le \Delta$ for all neighboring datasets $D \sim D'$, then the token-generation mechanism defined by the temperature-scaled softmax $\pi_i^\theta(w \mid h)=\frac{\exp(\ell_i(w,h)/T_i)}{\sum_{u\in\mathcal V_i}\exp(\ell_i(u,h)/T_i)}$ satisfies $\varepsilon_k$-DP at each token step with $\varepsilon_k \le \frac{2\Delta}{T_i}$.
\label{prop:token_bound}
\end{proposition}
\begin{proof}
Let $Z_D=\sum_{u\in\mathcal V_i}\exp(\ell_{i,D}(u,h)/T_i)$ denote the
normalization constant.
The assumption
$|\ell_{i,D}(w,h)-\ell_{i,D'}(w,h)|\le \Delta$ implies
$$
e^{-\Delta/T_i}
\le
\frac{\exp(\ell_{i,D}(w,h)/T_i)}
{\exp(\ell_{i,D'}(w,h)/T_i)}
\le
e^{\Delta/T_i},
$$
which yields $e^{-\Delta/T_i} \le Z_D/Z_{D'} \le e^{\Delta/T_i}$.
Combining them gives $\varepsilon_k \le 2\Delta/T_i$.
\end{proof}

\begin{corollary}[Message-Level Privacy Bound]
Assume that the logit mapping $\ell_i(w,h)$ satisfies the $\sup_{w,h} |\ell_{i,D}(w,h)-\ell_{i,D'}(w,h)| \le \Delta$ for all neighboring datasets $D \sim D'$. Then the response generation mechanism satisfies $\varepsilon$-DP at the message level with $\varepsilon \le \frac{2\Delta}{T_i} L_i^t$. 
\label{cor:message_bound}
\end{corollary}
\begin{proof}
The proof follows directly from Lemma \ref{lem:DP_sum} and Proposition \ref{prop:token_bound}.
\end{proof}

%% file: problem.tex
\section{Optimal DP Design of LLM-Agent}

Since stronger privacy protection typically degrades utility, parameters such as temperature and message length shape the output distribution and thus jointly influence both utility and privacy.

\subsection{Message-Level Utility}
The utility of a generated response for agent $i$ can be described by a function $\nu_i : \mathcal X_i \times \mathbb N \mapsto \mathbb R$. For the message $m_i^t$ generated at time $t$, the utility (information score) is $\nu_i(m_i^t,L_i^t)$, where the value of a response depends both on the token sequence $m_i^t$ and on the length $L_i^t$ of the generated message.

Recall that the language model induces a probability distribution over messages through the token-generation mechanism. Let $\mathcal M_{i,L_i^t}$ denote the set of all token sequences of length $L_i^t$ constructed from the vocabulary $\mathcal V_i$. For a fixed message length $L_i^t$, the probability that agent $i$ generates the message $m_i^t \in \mathcal M_{i,L_i^t}$ under temperature parameter $T_i$ is
$$\pi_i^{T_i,L_i^t}(m_i^t)= \frac{ \exp\!\left( \frac{1}{T_i} \sum_{k=1}^{L_i^t} \ell_i(w_{i,k}^t,h_{i,k}^t) \right)}{Z_{i,L_i^t}(T_i)},
$$ where $Z_{i,L_i^t}(T_i)$ is the normalization constant.
Then, the expected utility for responses of length $L_i^t$ generated by agent $i$ at temperature $T_i$ is
\begin{equation}
E_{i,L_i^t}(T_i) = \mathbb E_{\pi_i^{T_i,L_i^t}} \!\left[\nu_i(m_i^t,L_i^t)\right]. 
\label{eq:expected_utility_message}
\end{equation}
Assume that the utility of the message-level $\nu_i(\cdot,L_i^t)$ is bounded on $\mathcal M_{i,L_i^t}$, then for all $T_i>0$, $E_{i,L_i^t}(T_i)$ is also bounded, and is continuous on $T_i \in (0,\infty)$.

\begin{proposition}[Derivative and Monotonicity]
\label{prop:EU_derivative_monotone}
Fix $i,t$, $L_i^t$ and then define the cumulative logit score $U_i^t(m)\coloneqq \sum_{k=1}^{L_i^t}\ell_i(w_{i,k},h_{i,k}^t)$. Assume that the message distribution admits the Gibbs form $\pi_i^{T_i,L_i^t}(m)=\exp(U_i^t(m)/T_i)/Z_{i,L_i^t}(T_i)$. Then $E_{i,L_i^t}(T_i)$ is differentiable on $(0,\infty)$ and satisfies
\begin{equation}\label{eq:temperature_derivative}
\frac{dE_{i,L_i^t}}{dT_i}=\frac{-1}{T_i^2} \operatorname{Cov}_{\pi_i^{T_i,L_i^t}} \!\big(\nu_i(m,L_i^t),\,U_i^t(m)\big).
\end{equation}
If the covariance is nonnegative, then $dE_{i,L_i^t}/dT_i\le 0$ and $E_{i,L_i^t}(T_i)$ is nonincreasing in $T_i$.
\end{proposition}

\begin{proof}
For brevity, we omit the superscript and subscript.
\begin{align*}
E'(T)
&= \sum_{m\in\mathcal M}\nu(m)\pi_T'(m)\\
&= \frac{-1}{T^2}\sum_{m\in\mathcal M}\nu(m)\pi_T(m) \big(U(m)-\mathbb E_{\pi_T}[U]\big)\\
&= \frac{-1}{T^2}\big(\mathbb E_{\pi_T}[\nu U] -\mathbb E_{\pi_T}[\nu]\mathbb E_{\pi_T}[U]\big)\\
&= \frac{-1}{T^2}\operatorname{Cov}_{\pi_T}(\nu,U).
\end{align*}
The monotonicity follows immediately. 
\end{proof}

\subsection{Optimal Privacy-Utility Tradeoff}
In many practical systems, the privacy requirement is not a strict threshold but rather a design knob that should be balanced against response quality. In such cases, it is convenient to replace the hard constraint by a regularized objective that rewards stronger privacy protection. For example, a simple proxy is to reward higher temperature in the objective. Given a privacy preference parameter $\lambda_i \ge 0$, the \emph{regularized privacy-utility design problem} can be formulated as
\begin{equation}
    \max_{T_i>0}\;E_{i,L_i^t}(T_i)+\frac{\lambda_i}{L_i^t}T_i.
    \label{eq:opt_problem}
\end{equation}
\begin{proposition}[Optimal Temperature]
\label{thm:optimal_temperature}
Consider the regularized privacy-utility design problem defined in \ref{eq:opt_problem}, with the derivative of the expected utility satisfying \eqref{eq:temperature_derivative}. Then any interior optimal temperature $T_i^\ast>0$ satisfies the first-order condition $\frac{d}{dT_i}E_{i,L_i^t}(T_i^\ast)+\frac{\lambda_i}{L_i^t}=0$. Using \eqref{eq:temperature_derivative}, the optimal temperature is characterized by $$\frac{\lambda_i}{L_i^t}=\frac{1}{(T_i^\ast)^2}\operatorname{Cov}_{\pi^{T_i^\ast,L_i^t}_i}\big(\nu_i(m,L_i^t),U_i^t(m)\big).$$
\end{proposition}
\begin{proof}
    The proof follows Equation \eqref{eq:temperature_derivative}.
\end{proof}

%% file: experiment.tex
\section{Empirical Case Study}
We illustrate the proposed privacy-utility temperature selection framework using a generative language model. The objective is to show how the sampling temperature emerges from the interaction between model probabilities, logit score, and information score.

\subsection{Experimental Setup}
We prompt GPT-2 with the query: ``You are given a cybersecurity incident database. Incident records: $D$. Question: Based on the database, the most frequent attack type is''. Another prompt is similar except that $D$ is replaced with a neighboring dataset $D'$.
\begin{figure*}\vspace{+2mm}
    \centering
\begin{subfigure}[t]{0.32\textwidth}
    \includegraphics[width=\textwidth]{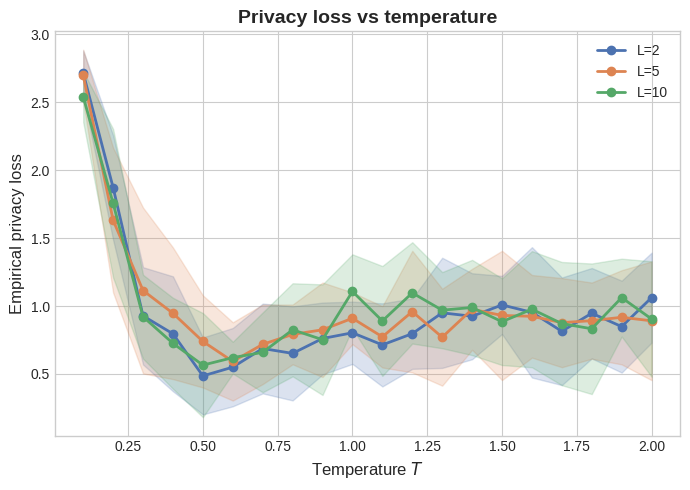}
    \caption{Empirical privacy loss.}
    \label{fig:privacy_T}
\end{subfigure}
\hfill
\begin{subfigure}[t]{0.32\textwidth}
    \includegraphics[width=\textwidth]{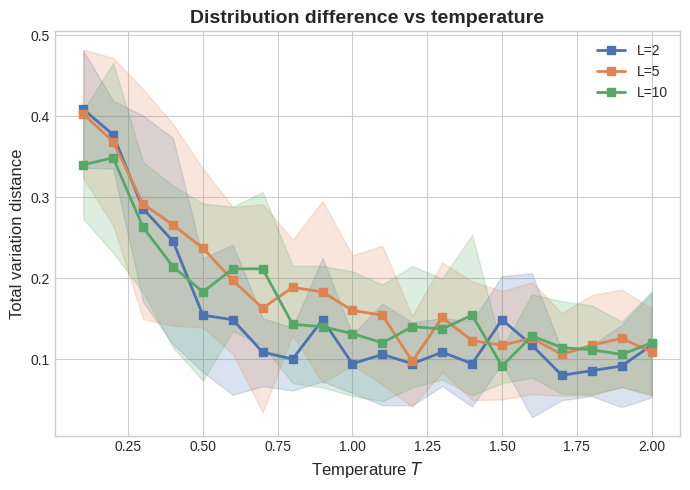}
    \caption{Total variation distance.}
    \label{fig:TV_T}
\end{subfigure}
\hfill
\begin{subfigure}[t]{0.32\textwidth}
    \includegraphics[width=\textwidth]{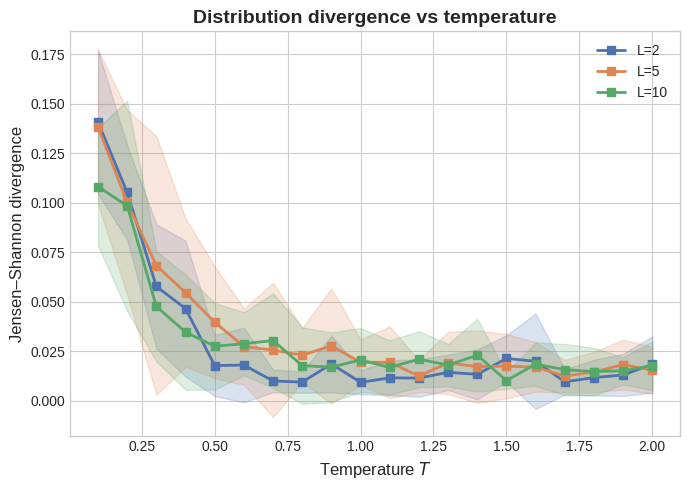}
    \caption{Jensen-Shannon divergence.}
    \label{fig:JS_T}
\end{subfigure}
\vspace{-0mm}\caption{Privacy leakage under different temperature. Lines show the means while shaded areas indicate the standard deviations.}\vspace{-0mm}
\label{fig:privacy}
\end{figure*}
\begin{figure*}
    \centering
\begin{subfigure}[t]{0.32\textwidth}
    \includegraphics[width=\textwidth]{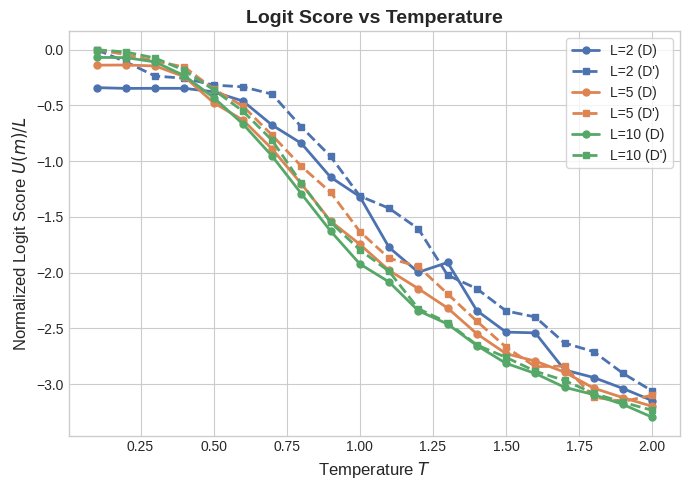}
    \caption{Cumulative logit score.}
    \label{fig:U_T}
\end{subfigure}
\hfill
\begin{subfigure}[t]{0.32\textwidth}
    \includegraphics[width=\textwidth]{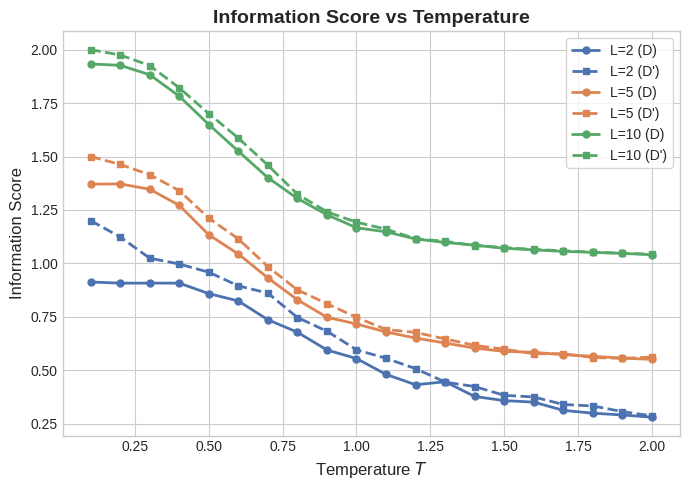}
    \caption{Information score.}
    \label{fig:IS_T}
\end{subfigure}
\hfill
\begin{subfigure}[t]{0.32\textwidth}
    \includegraphics[width=\textwidth]{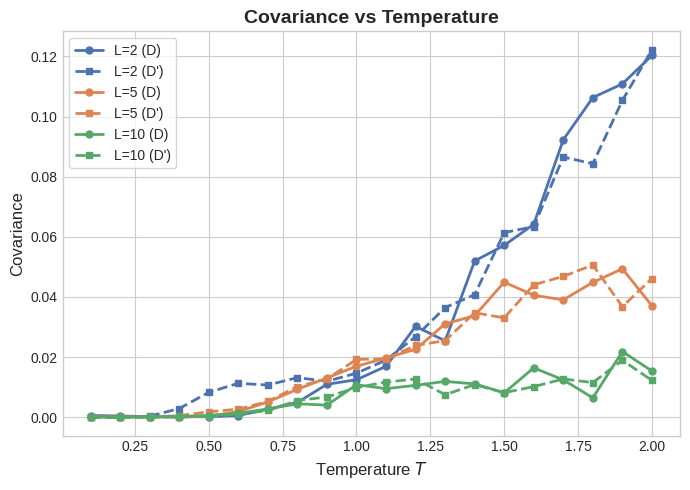}
    \caption{Covariance.}
    \label{fig:cov_T}
\end{subfigure}
\vspace{-0mm}\caption{The quantities in the proposed privacy–utility framework under different temperatures.}\vspace{-0mm}
\label{fig:measure}
\end{figure*}
We set length $L=2, 5,10$ and sample $250$ responses of each length with temperature $T$ from $0.1$ to $2.0$. 
\subsection{Metrics}
To empirically evaluate privacy leakage, we compare the output distributions of the model from neighboring datasets $D$ and $D'$ using several distributional metrics. Given sampled outputs mapped to a finite label space $\mathcal{Y}$, we first construct smoothed empirical distributions using Laplace smoothing:
$$
\hat{P}(y)=\frac{c_P(y)+\alpha}{n+\alpha|\mathcal{Y}|},\qquad
\hat{Q}(y)=\frac{c_Q(y)+\alpha}{n+\alpha|\mathcal{Y}|},
$$
where $c_P(y)$ and $c_Q(y)$ denote the counts of label $y$ in samples generated from datasets $D$ and $D'$, $n$ is the number of samples, and $\alpha>0$ prevents zero probabilities. 

Using these distributions, we estimate the smoothed empirical privacy loss $\hat{\varepsilon}=\max_{y\in\mathcal{Y}}\Big|\log(\hat{P}(y)/\hat{Q}(y))\Big|,$ and show the results in Figure \ref{fig:privacy_T}.
To complement this measure, we also compute the total variation distance $\mathrm{TV}(\hat{P},\hat{Q})=\frac{1}{2}\sum_{y\in\mathcal{Y}}\big|\hat{P}(y)-\hat{Q}(y)\big|,$ the Jensen-Shannon divergence $\mathrm{JS}(\hat{P},\hat{Q})=0.5\mathrm{KL}(\hat{P}\|M)+0.5\mathrm{KL}(\hat{Q}\|M),$ where $M=\tfrac{1}{2}(\hat{P}+\hat{Q})$, $\mathrm{KL}(\cdot\|\cdot)$ denotes the Kullback-Leibler divergence, and show them in Figure \ref{fig:TV_T} and \ref{fig:JS_T}, respectively.

Figure \ref{fig:measure} illustrates how the quantities in the proposed privacy–utility framework vary with temperature. The cumulative logit score is defined in Proposition \ref{prop:EU_derivative_monotone}, and the information score is defined as $\nu_i(m,L)=f_i(U_i(m),L)=e^{U_i(m)}+0.1L$, which is nondecreasing and convex with respect to the cumulative logit score variable.

\subsection{Discussion}
In Figure \ref{fig:privacy}, Across all message lengths, the empirical privacy loss, total variation distance, and Jensen–Shannon divergence generally decrease as the temperature increases. This trend indicates that higher temperatures reduce the sensitivity of the output distribution to changes in the underlying dataset, hence, may offer better privacy protection. Intuitively, increasing the temperature flattens the softmax distribution, which reduces the influence of individual logits and mitigates the impact of dataset differences on the generated outputs.


In Figure \ref{fig:measure}, the cumulative logit score decreases with increasing temperature. This behavior reflects the fact that higher temperatures encourage sampling from lower-probability tokens, reducing the average model confidence of generated responses. The information score exhibits a similar decreasing trend.
The covariance between the information score and the cumulative logit score remains nonnegative across the temperature range. According to the monotonicity result in Proposition \ref{prop:EU_derivative_monotone}, a nonnegative covariance implies that the expected information score is nonincreasing with respect to temperature. The empirical results here, therefore, provide consistent evidence.
In addition, longer responses accumulate more logit and information scores, which may slightly increase both privacy leakage and utility variation. This observation aligns with the analysis showing that privacy loss can grow with the number of generated tokens.

%% file: conclusion.tex
\section{Conclusion and Future Work}

This paper develops a principled framework for studying the privacy properties of AI agents powered by large language models. By modeling response generation as a stochastic mechanism over token sequences, we introduce differential privacy notions at both the token and message levels. The analysis establishes privacy bounds that relate privacy leakage to generation parameters such as temperature and message length. These results provide theoretical insight into how decoding strategies influence the sensitivity of model outputs to underlying datasets.
Building on the theoretical characterization, we then formulate an optimal privacy–utility design problem in which temperature serves as a controllable parameter for balancing response quality and privacy protection. Empirical experiments with GPT-2 illustrate the theoretical trends and show that increasing temperature generally reduces privacy leakage while lowering utility.
Future work includes extending the framework to multi-agent settings, incorporating more realistic utility functions, jointly optimizing temperature, message length, and other privacy-related control parameters, as well as conducting experiments on enterprise datasets and business-oriented LLMs.